\documentclass[%
 reprint,
 amsmath,amssymb,
 aps,
]{revtex4-1}

\usepackage{graphicx}
\usepackage{dcolumn}
\usepackage{bm}
\usepackage{physics}
\usepackage{multirow}
\usepackage{hyperref}
\usepackage{amsthm}
\usepackage{apptools}
\usepackage{chngcntr}
\AtAppendix{\counterwithin{Lemma}{section}}

\newtheorem{Theorem}{Theorem}
\newtheorem{Lemma}{Lemma}

\begin{document}

\preprint{APS/123-QED}

\title{Transversal switching between generic stabilizer codes}

\author{Cupjin Huang$^{1,}$}
\email{cupjinh@umich.edu}
\author{Michael Newman$^{2,}$}%
 \email{mgnewman@umich.edu}
\affiliation{%
 $^1$Department of Electrical Engineering and Computer Science\\
  \smallskip
  $^2$Department of Mathematics\\
  \smallskip
  University of Michigan, Ann Arbor, MI 48109, USA
}%

\date{\today}

\begin{abstract}
    We propose a randomized variant of the stabilizer rewiring algorithm (SRA), a method for constructing a transversal circuit mapping between any pair of stabilizer codes.  As gates along this circuit are applied, the initial code is deformed through a series of intermediate codes before reaching the final code.  With this randomized variant, we show that there always exists a path of deformations which preserves the code distance throughout the circuit, while using at most linear overhead in the distance.  Furthermore, we show that a random path will almost always suffice, and discuss prospects for implementing general fault-tolerant code switching circuits.
\end{abstract}

\maketitle


\section{Introduction}

It is an oft-cited fact that no quantum error-correcting code can implement a universal transversal logical gate set \cite{Eastin:2008, Zeng:2007, Newman:2017}.  As a result, there have been several attempts to circumvent this no-go theorem to achieve universal fault-tolerant quantum computation.  These candidates include magic state distillation \cite{Fowler:2013, Bravyi:2012}, gauge fixing \cite{Paetznick:2013, Bombin:2013}, and more recently pieceable fault-tolerance \cite{Yoder:2016, Yoder:2017}.  These last two candidates can be seen as a special case of the more general approach of code switching \cite{Hill:2013, Anderson:2014, Brun:2015, Bombin:2007, Nautrup:2017}.

Code switching is a natural idea: given two codes, map information encoded in one code to information encoded in the other.  For this mapping to be fault-tolerant, we must often perform several intermediate error-correction steps to ensure that faults do not grow out of hand.  Thus, it is essential that during a circuit switching between codes, the extremal error-correcting codes are deformed through a series of intermediate error-correcting codes from one to another.  This notion of intermediate error-correction was used in \cite{Anderson:2014} to implement universal transversal computation by switching between the Steane and Reed-Muller codes, whose complementary transversal gate sets are universal when taken together.  However, universal fault-tolerant computation is not the only consideration in choosing error-correcting codes, and different codes can be tailored to different tasks.  For this reason, it would be nice to have a way of converting between different quantum codes fault-tolerantly.

Simply decoding and re-encoding information is undesirable, since the bare information becomes completely unprotected during this transformation.  Past work has succeeded in constructing fault-tolerant circuits for switching between particular quantum error-correcting codes fault-tolerantly, while providing guarantees that these circuits are optimal within some framework \cite{Hill:2013}.

Recently, \cite{Colladay:2017} considered switching between generic stabilizer codes, and proposed the stabilizer rewiring algorithm (SRA) for constructing a transversal circuit mapping between \emph{any} pair of stabilizer codes.  The circuit complexity scales quadratically with the code length, and depends on a choice of presentation for the code generators.  Different presentations will result in different circuits mapping between different sets of at most $n$ intermediate codes.  This circuit necessarily fails to be fault-tolerant when these intermediate codes have low distance. This leads to the central question: \emph{is there an efficient way of fault-tolerantly converting between generic stabilizer codes?}

\subsection{Results}

Towards this goal, we propose a randomized variant of the SRA, the randomized SRA (rSRA). We show that for any pair of stabilizer codes, with at most linear overhead with respect to the distance of the codes, there always exists a transversal circuit that maps between intermediate codes of high distance.  Furthermore, using slightly more overhead, such a path can be found with high probability.  In particular, we show the following.

\newtheorem*{thm:1}{Theorem \ref{thm:1}}
\begin{thm:1}[Informal] For any two $[[n,k,d]]$ stabilizer codes $S_1$ and $S_2$, the rSRA scheme gives a transversal circuit mapping from $S_1$ to $S_2$ where each intermediate code has distance at least $d$ with probability $1-\varepsilon$, using
    $$m=O\left(d\log\frac{n}{d}+\log\frac1\varepsilon\right)$$
    ancilla qubits.
\end{thm:1}

This \emph{distance-preserving} property is a necessary, but not sufficient condition to ensure a fault-tolerant mapping.  So while the algorithm does \emph{not} necessarily yield a fault-tolerant conversion, it gives a universal upper bound on the number of ancilla qubits required for distance-preserving transversal code transformation. As was noted in \cite{Colladay:2017}, the usefulness of this scheme is in its generality.  While the upper bound may be of independent conceptual interest, we hope that with modification, the rSRA can be applied as a useful schema for searching for fault-tolerant paths between small codes.  We provide small examples of such transversal paths in Section~\ref{sec:examples}, including a path between the $[[5,1,3]]$ and $[[7,1,3]]$ codes that without modification protects against erasure with no overhead.

\subsection{Organization}
In Section~\ref{sec:prelim} we introduce some preliminaries and notation that will be used throughout the paper. The main rSRA schematic is presented in Section~\ref{sec:rSRA}. Some examples illustrating different parts of the rSRA and demonstrating small conversion paths are presented in Section~\ref{sec:examples}. The proof of the main theorem is presented in Section~\ref{sec:proof}. Some discussion on fault-tolerance and possible improvements for the rSRA can be found in Section~\ref{sec:discussion}. Readers interested in the technical details may refer to the appendices for additional lemmas.

\section{Preliminaries}
\label{sec:prelim}
Let $\mathcal{P}^n$ denote the $n$-qubit Pauli group.  Then a \emph{stabilizer group} $S \subseteq \mathcal{P}^n$ is an abelian subgroup of the Pauli group not containing $-I$.  To any such stabilizer group $S$, we can associate a subspace $C_S \subseteq (\mathbb{C}^2)^{\otimes n}$ defined as the simultaneous $+1$-eigenspace of all the operators in $S$.  We call such a subspace $C_S$ a \emph{stabilizer code}.

A stabilizer code $C_S$ has parameters $[[n,k,d]]$.  Here, $n$ is the number of physical qubits comprising the code, $k$ is the number of logical qubits of the code $\log(\dim(C_S))$, and $d$ is the distance of the code.  More precisely, the normalizer $\mathcal{N}_{\mathcal{P}^n}(S)$ represents the set of logical Pauli operators for $C_S$, and so $$d := \min\limits_{L \in \mathcal{N}(S) \backslash S}(|L|)$$ where $|\cdot|$ denotes the weight of the Pauli operator.  Note that the number of stabilizer in the corresponding stabilizer subgroup is $n-k$.

Given any stabilizer group $S$, if we choose a generating set $G_S$ for $S$, we can define a syndrome map
\begin{align*}
Syn_{G}&: \mathcal{P}^n \longrightarrow \{0,1\}^{n-k} \\
Syn_G&(e)_i = \begin{cases} 0 \text{ if $[e,g_i] = 0$} \\ 1 \text{ if $\{e,g_i\} = 0$} \end{cases}
\end{align*}
\noindent for $G = (g_1, \ldots g_{n-k})$.  Then equivalently, $$d = \min\limits_{L \in \ker(Syn_G) \backslash S}(|L|)$$ and is independent of the choice of $G$.

Another convenient formalism for describing stabilizer groups is as subspaces of symplectic vector spaces, and we will use the two formulations interchangeably.  For any $P \in \mathcal{P}^n/\mathcal{U}(1)$, if $$ P = X^{a_1}Z^{b_1} \otimes X^{a_2}Z^{b_2} \ldots \otimes X^{a_n}Z^{b_n}$$ then we can associate to $P$ the vector $\vec{P} := (\vec{a} | \vec{b})^T \in \mathbb{F}_2^{2n}$.  Equip $\mathbb{F}_2^{2n}$ with a symplectic bilinear form $$\langle \vec{v},\vec{w} \rangle := \vec{v}^TB\vec{w}$$ where $B$ is the $2n \times 2n$ block matrix defined by $$B = \begin{pmatrix} 0 & I \\ I & 0 \end{pmatrix}. $$  Then Paulis $P,Q$ commute if any only if their associated vectors $\vec{P},\vec{Q}$ are orthogonal in this vector space.  Thus, we can equivalently define a stabilizer group as a self-orthogonal subspace of this vector space.  A \emph{generator matrix} $G$ is then a choice of basis for this subspace, so that for $C$ an $[[n,k]]$ code, $G$ will be a rank $(n-k)$ matrix of shape $2n \times (n-k)$.  The syndrome map can then be similarly defined as $$Syn_G(\vec{P}) = G^TB\vec{P}.$$  Further note that for any $A \in GL(\mathbb{F}_2,n-k)$, for any generator matrix $G$ for $S$, $GA^T$ is also a generator matrix for $S$.  The syndrome map satisfies $$Syn_{GA^T}(\vec{P}) = (GA^T)^TB\vec{P} = AG^TB\vec{P} = A\cdot Syn_G(\vec{P}).$$

So any action on the generator matrix induces a corresponding action on the syndrome vectors themselves.

Finally, we call a circuit $C$ on a class of encoded inputs \emph{$t$-fault-tolerant} if it is $t$-fault-tolerant in the exRec formalism \cite{Aliferis:2006}. Formally, given error correction procedure $EC$, $C$ is $t$-fault-tolerant if for any choice of $t$ faulty components in the combined circuit $EC \cdot C \cdot EC$, a faultless version of $EC$ applied to the output of the combined circuit can successfully recover the data. If $t \geq 1$ we may simply call the circuit fault-tolerant.

\section{The \lowercase{r}SRA schematic}
\label{sec:rSRA}

The rSRA modifies the SRA presented in \cite{Colladay:2017}, whose central insight is the following.  Consider two stabilizer groups $S,S'$ with generating sets $G,G'$ satisfying the following nice property:
\begin{align*}
G &= \{g, g_1 , \ldots, g_l\} \\
G' &= \{g', g_1, \ldots, g_l\}
\end{align*}
where $\{g, g'\} = 0$.  We call two such codes for which one can choose such generating sets \emph{adjacent}.  Then one can readily check that the Clifford gate $\frac{1}{\sqrt{2}}(1 + g'g)$ maps information encoded in the stabilizer code defined by $G$ to the same information encoded in the stabilizer code defined by $G'$.  Letting $\ket{\psi}_G$ denote a logical state in the code associated to $G$, we see that
\begin{align*}
g_i \cdot \frac{1}{\sqrt{2}}(1 + g'g)\ket{\psi}_G &= \frac{1}{\sqrt{2}}(1 + g'g) \ket{\psi}_G \text{, and} \\
g' \cdot \frac{1}{\sqrt{2}}(1 + g'g)\ket{\psi}_G &=\frac{1}{\sqrt{2}}(g' + g)\ket{\psi}_G \\
&= \frac{1}{\sqrt{2}}(g' + 1) \ket{\psi}_G \\
&= \frac{1}{\sqrt{2}}(1 + g'g)\ket{\psi}_G.
\end{align*}
The insight is that this mapping can be done transversally.  While the Clifford transformation described need not be transversal, it can be simulated by a transversal Pauli measurement supplemented by a transversal Pauli gate controlled on classical information.  This is similar to gauge-fixing, in which one measures a logical operator of the gauge and then applies a corresponding logical gauge operator conditioned on the outcome.  To see this, consider the circuit described by:

\begin{enumerate}
\item Measure $g'$.
\item Apply $g$ conditioned on measurement outcome $-1$.
\end{enumerate}

\noindent Let $P^{\pm}$ denote the projector onto the ${+1}/{-1}$ eigenspace of $g'$.  Then, if the measurement outcome is $+1$, $$ \frac{1}{\sqrt{2}}(1 + g'g) \ket{\psi}_G= \frac{1}{\sqrt{2}}(1+g')\ket{\psi}_G = \sqrt{2}P^+\ket{\psi}_G .$$  Furthermore,

\noindent If the measurement outcome is $-1$,

\begin{align*}
\frac{1}{\sqrt{2}}(1 + g'g)\ket{\psi}_G &= \frac{1}{\sqrt{2}}(g-gg')\ket{\psi}_G \\
&= \frac{1}{\sqrt{2}}g(1-g')\ket{\psi}_G \\
&= \sqrt{2}gP^-\ket{\psi}_G.
\end{align*}

\noindent Thus, we see that we can \emph{transversally} perform the mapping $\ket{\psi}_G \rightarrow \ket{\psi}_{G'}$.

Now consider the more general case in which we have (non-adjacent) $S,S'$ describing $[[n,k]]$ and $[[n',k]]$ codes respectively.  We now describe a general randomized algorithm for outputting a circuit switching between these two codes, similar to \cite{Colladay:2017}, and will later show that this circuit is distance-preserving with high probability.  The inputs are arbitrary generator matrices $G,G'$ for stabilizer groups $S,S'$, along with a choice of ancilla size $m \in \mathbb{N}$.

\subsection{Preparing the generator matrices}

\begin{enumerate}
\item Append $\ket{0}$ ancilla to the smaller code so that the codes are of equal size.  We now assume that both codes are $[[n,k]]$ codes.
\item Append $\ket{0}^{\otimes m}$ to the first code, and $\ket{+}^{\otimes m}$ to the second. Note that this is equivalent to defining a pair of new stabilizer codes
    $$\hat{S}=\langle S\otimes I^{\otimes m},I^{\otimes n}\otimes Z\otimes I^{\otimes m-1},\ldots, I^{\otimes n+m-1}\otimes Z \rangle,$$
    $$\hat{S}'=\langle S'\otimes I^{\otimes m},I^{\otimes n}\otimes X\otimes I^{\otimes m-1},\ldots, I^{\otimes n+m-1}\otimes X \rangle.$$
\item Choose $G_A = G_A'$ to be a basis for the subspace defined by $\hat{S} \cap \hat{S}'$.
\item Choose $G_B$ to extend the basis of $G_A$ to a basis for $\mathcal{N}(\hat{S}') \cap \hat{S}$ and choose $G_B'$ to extend the basis of $G_A$ to a basis for $\mathcal{N}(\hat{S}) \cap \hat{S}'$.
\item\label{5} Choose $G_C$ to extend the basis $G_A \cup G_B$ to a basis for $\hat{S}$ and $G_C'$ to extend the basis $G_A' \cup G_B'$ to a basis for $\hat{S}'$.
\item\label{6} Let $H$ be the \emph{commutativity matrix} for $G_C,G_C'$ defined by $H := {G_C'}^TBG_C$. By Lemma~\ref{lem:2}, $H$ is invertible with dimension $|G_C|\times |G_C|$, where $|G_C|\geq m$.   So we can choose $M,N \in GL(\mathbb{F}_2,|G_C|): M^THN = I_{|G_C|}$ and redefine
\begin{align*}
G_C &\leftarrow G_C\cdot M \\
G_C' &\leftarrow G_C'\cdot N.
\end{align*}
\item\label{7} Choose uniformly at random $V,V' \in_r \mathbb{F}_2^{|G_C| \times |G_B|}$ and a $U \in_r GL(\mathbb{F}_2, |G_C|)$.
\item\label{8} Redefine
\begin{align*}
G_C^T &\leftarrow U(VG_B^T + G_C^T) \\
{G_C'}^T &\leftarrow {(U^{-1})}^T(V'{G_B'}^T + {G_C'}^T) \\
\end{align*}
Note that this does not change the commutativity matrix since $$U(VG_B^T + G_C^T) B (G_C'+G'_B{V'}^T)U^{-1} = I_{|G_C|}.$$
\item\label{9} Let $G_B = \{g_1, \ldots, g_{|G_B|}\}$ and $G_B' = \{g_1',\ldots, g_{|G'_B|}\}$.  For each $g_i \in G_B$, choose $\overline{g_i}$ satisfying
\begin{align*}
[\overline{g_i},G_A] &= 0 \\
[\overline{g_i},G_C] &= 0 \\
[\overline{g_i},G_C'] &= 0 \\
[\overline{g_i},\{g_{i+1},\ldots,g_{|G_B|}\}] &= 0 \\
[\overline{g_i}, \{g_{i+1}', \ldots, g_{|G_B|}'\}] &= 0 \\
[\overline{g_i}, \{\overline{g_{1}}, \ldots, \overline{g_{i-1}}\}] &= 0 \\
\{\overline{g_i},g_i\} &= 0 \\
\{\overline{g_i},g_i'\} &= 0.
\end{align*}
To see that such a choice of $\overline{g_i}$ always exists, note that it must satisfy at most $2n$ affine linear equations, all of which are linearly independent, in a space of dimension $2n$.

Now that we have prepared the generator matrices, we will step-by-step map between adjacent codes transversally.

\subsection{Applying the transformation}

\item\label{10} For $1 \leq i \leq |G_B|$ indexing the elements of $G_B$, perform the transformation $g_i \mapsto \overline{g_i}$.  Note that the resulting stabilizer codes are adjacent, and so the preceding discussion gives a transversal circuit for each mapping.

\item\label{11} For $1 \leq i \leq |G_C|$ indexing the elements of $G_C$, perform the transformation $g_i \mapsto g_i'$.  Again, since the codes are adjacent, the mapping can be done transversally.

\item\label{12} For $1 \leq i \leq |G_B|$ indexing the elements of $G_B'$, perform the transformation $\overline{g_i} \mapsto g_i'$ starting from $i = |G_B|$ and working backwards towards $i=1$.  Again, we have a transversal circuit for each mapping.

\item\label{13} Discard the ancilla.
\end{enumerate}

This randomized variant differs from the original SRA in several ways.  First, there is the introduction of ancilla, which we will see are vital for preserving the distance.  Next, the SRA fixes the generating sets $G,G'$ subject to the same $G_A$ and $G_C$ conditions, but with different $G_B$ conditions.  Namely, the SRA fixes the $\overline{g}$ to be the product of the complementary logical operators to those operators in $G_B$ and $G_B'$, which can be seen as nontrivial logical
operators on the opposite code.  This allows for a certain degree of freedom in choosing the order in which one converts between the two codes, but restricts the $G_C$,$G_C'$ that are available to use.  Also in the SRA, only the set of valid permutations among $G_B$ and $G_C$ are considered, which restricts the search for a distance-preserving mapping. In the rSRA, we consider the full set of invertible transformations on $G_C$ for a better chance of success.  Finally, the transformation described above is \emph{symmetric} in the sense that switching from $G$ to $G'$ or $G'$ to $G$ after step~\ref{9} results in the same set of intermediate codes.  We will see that this simplifies the set of errors we must consider.

\section{Distance-preservation for small codes}
\label{sec:examples}

We have now described a way of constructing a transversal circuit mapping information encoded in $G$ to information encoded in $G'$ through the use of Shor-style measurement (see Appendix \ref{Shor}).

However, we have no \emph{a priori} guarantee that these intermediate codes, resulting from the sequence of deformations, will themselves be error-correcting.  In light of this, we offer several examples of small distance-preserving circuits generated from the rSRA.  These illustrate the necessity of the aforementioned modifications, which are centered around choosing a path so that all of the intermediate codes have high distance.  In these examples, the extremal codes all have distance $3$, and so we call the circuit distance-preserving if the intermediate codes all have distance $\geq 3$.
\subsection{$[[7,1,3]] \longleftrightarrow [[5,1,3]]$}

With $m=0$, one can generate a distance-preserving map from the $[[7,1,3]]$ Steane code to the perfect $[[5,1,3]]$ code using the rSRA with $17$ multi-qubit gates.  An optimal fault-tolerant (and so distance-preserving) transformation using $CZ$ gates between these two codes was found via brute force search in \cite{Hill:2013} and involves $14$ multi-qubit gates.  The circuit output by the rSRA requires no overhead in data qubits compared to the three extra qubits required in \cite{Hill:2013}.  However, because the $[[5,1,3]]$ code is perfect, any conversion without ancilla must only be able to protect against erasure, for reasons detailed in Section~\ref{sec:discussion}.  Note also that there must be conversions with large separation between the circuit provided by the rSRA and the optimal fault-tolerant circuit, in particular when $G$ and $G'$ are locally unitarily equivalent.
\\
\begin{table}[htb]
 \centerline{\begin{tabular}{|c|c|c|}
    \hline
   Type & $[[7,1,3]]$ & $[[5,1,3]]$\\
    \hline\hline
  $G_A$ & $-YXXYIZZ$ & $-YXXYIZZ$\\
  \hline
  \multirow{5}{*}{$G_C$} & $ZZZZIII$ & $IXZZXII$ \\
  \cline{2-3}
  & $-YYXXZZI$ & $XZZXIII$\\
  \cline{2-3}
  & $-IXZYYZX$ & $XIXZZII$ \\
  \cline{2-3}
  & $-XIYZYZX$ & $ZXIXZII$ \\
  \cline{2-3}
  & $-ZYYZIXX$ & $IIIIIZI$ \\
  \hline
  \end{tabular}}
  \caption{The generator matrices defining a distance-preserving conversion, proceeding from top to bottom.  We follow steps~\ref{10}~-~\ref{13} of the algorithm.}
\end{table}
\subsection{$(34) \cdot [[7,1,3]] \longleftrightarrow [[9,1,3]]$}
With $m = 0$, one can convert from the $(34)$ permutation of the $[[7,1,3]]$ Steane code to Shor's $[[9,1,3]]$ code while preserving the distance.  However, for the standard choice of generator matrices, no permutation on the ordering of the deformations will suffice.  Thus, we must choose $U \in GL(\mathbb{F}_2, |G_C|)$ rather than restricting $U$ to be a permutation matrix.  A choice of generator matrices for which the circuit is distance-preserving is presented below.
\begin{table}[h]
 \centerline{\begin{tabular}{|c|c|c|}
    \hline
   Type & $(34) \cdot [[7,1,3]]$ & $[[9,1,3]]$\\
    \hline\hline
  \multirow{2}{*}{$G_A$} & $ZZIIZZIII$ & $ZZIIZZIII$\\
  \cline{2-3}
   & $IIIIIIIZZ$ & $IIIIIIIZZ$ \\
  \hline
  $G_C$ & $YIIYYIYII$ & $-YXYZZIXXX$ \\
  \hline
 $G_B$ & $ZZZZIIIZI$ & $ZZIIIIZZI$ \\
  \hline
  \multirow{4}{*}{$G_C$}& $-ZZYYXXIII$ & $IZZZZIIII$\\
  \cline{2-3}
  & $ZIIZZIZII$ & $-YYXXXXIII$ \\
  \cline{2-3}
  & $-XZZXYIYII$ & $-XYYIIIXXX$ \\
  \cline{2-3}
  & $-IYZXZXYII$ & $IIIZZIIII$ \\
  \hline
  \end{tabular}}
  \caption{The conversion proceeds from top to bottom.  As the $G_B$ elements commute, we perform an intermediate conversion to the product of the complementary logical operators, which in this case are $XXXXXXXXX$ and $XXXXXXXII$ respectively.  This small modification is similar to the SRA \cite{Colladay:2017}, which we adopt here for ease of presentation.}
 \end{table}
\subsection{ $[[7,1,3]] \longleftrightarrow (34)\cdot[[7,1,3]]$}
For $m=0$, it was observed in \cite{Colladay:2017} that one cannot use the SRA to convert between the $[[7,1,3]]$ code, and the $(34)$ permutation of the $[[7,1,3]]$ code while preserving the distance.  In fact, there does not exist a $ U \in GL(\mathbb{F}_2,|G_C|)$ that allows the intermediate codes to be error-correcting.  In contrast, with $m = 2$, there does exist such a distance-preserving circuit, emphasizing the need for ancilla.  Moreover, brute force search shows that this is the minimal number of ancilla required to produce a distance-preserving circuit within this framework. However, note that qubit permutations are themselves automatically fault-tolerant by simply relabeling the wires, rather than applying a fault-tolerant physical SWAP gate.
\begin{table}[h]
 \centerline{\begin{tabular}{|c|c|c|}
    \hline
   Type & $[[7,1,3]]$ & $(34) \cdot [[7,1,3]]$\\
    \hline\hline
  \multirow{4}{*}{$G_A$} & $XXIIXXIII$ & $XXIIXXIII$\\
  \cline{2-3}
   & $ZZZZIIIII$ & $ZZZZIIIII$ \\
   \cline{2-3}
   & $ZZIIZZIII$ & $ZZIIZZIII$ \\
   \cline{2-3}
   & $XXXXIIIII$ & $XXXXIIIII$\\
  \hline
  \multirow{4}{*}{$G_C$}& $XIXIXIXII$ & $YIIYYIYXX$\\
  \cline{2-3}
  & $IIIIIIIZZ$ & $XIIXXIXIX$ \\
  \cline{2-3}
  & $ZIZIZIZIZ$ & $IIIIIIIXX$ \\
  \cline{2-3}
  & $YIYIYIYZZ$ & $XIIXXIXXX$ \\
  \hline
  \end{tabular}}
  \caption{The conversion proceeds from top to bottom.  In particular, we use $2$ extra ancilla qubits, for $9$ physical qubits in total.}
 \end{table}
\section{Distance bounds}
\label{sec:proof}

We now show that, with low overhead and high probability, the described rSRA will yield a distance-preserving circuit.  More specifically, we show that the intermediate codes preserve the distance of the extremal codes.

\begin{Theorem}\label{thm:1}
Let $S$,$S'$ be any two stabilizer codes with parameters $[[n_1,k,d_1]]$ and $[[n_2,k,d_2]]$, respectively.  Let $d = \min\{d_1, d_2\}$ and $n = \max\{n_1,n_2\}$. Then, the rSRA will output a distance-preserving circuit mapping information encoded in $S$ to information encoded in $S'$ with probability $1-\epsilon$ using $m=O(d\log\frac{n}{d}+\log \frac1\varepsilon) $ ancilla qubits.
\end{Theorem}

\begin{proof}
    Consider a particular error $e: |e|<d$. There are four different types of errors to consider. \vspace{0.3cm}
    
    \item[$(1)$] \underline{$e \in S \cap S'$}:  In this case, $e \in Span(G_A)$, and so remains passively corrected throughout the transformation. \\
    \item[$(2)$] \underline{$e \in S \setminus \mathcal{N}(S')$}:  In this case, we can decompose $e = g_A + g_B + g_C$ where $g_A \in Span(G_A), g_B \in Span(G_B),$ and $g_C \in Span(G_C)$.  Furthermore, $g_C \neq 0$, or else $e$ would be a logical operator of weight $<d$ for $S'$.  Thus, $e$ must be detected by $G_C'$, and so it remains detectable after step \ref{11}.  In particular, before the end of step \ref{11}, $e$ must fall out of the intermediate stabilizer group.  Suppose this occurs for the first time when transforming between two adjacent codes whose stabilizer groups differ by $g,g'$.  Then we can write $e = g + \sum_i a_ig_i$, and as $g'$ commutes with all other $g_i$, it must be that $\{e,g'\} = 0$.  Since $g'$ remains in each intermediate code up through step \ref{11}, $e$ must be detectable throughout. \\
    \item[$(3)$] \underline{$e \in S' \setminus \mathcal{N}(S)$}:  This error is just an error of type $(2)$ when performing the opposite transformation from $S'$ to $S$.  By symmetry of the scheme, the set of intermediate codes during this opposite transformation is the same, and so these errors remain detectable by the preceding argument. \\
    \item[$(4)$] \underline{$e \not \in  \mathcal{N}(S) \cup  \mathcal{N}(S')$}:  Let $G_C^{(0)}, G_{C}'^{(0)}$ be the bases $G_C$ and $G'_C$ we choose after step \ref{6} in the rSRA scheme, and let $G_C^{(1)}, G_C'^{(1)}$ be the bases we choose after step \ref{8}. Note that the syndrome map for $G_C^{(1)}$ can then be expressed as 
$$Syn_{G_C^{(1)}}(e)=U(V\cdot Syn_{G_B}(e)+Syn_{G_C^{(0)}}(e)).$$  In this case it must be that 
   \begin{align*} (Syn_{G_A}(e)| Syn_{G_B}(e)| Syn_{G_C^{(0)}}(e))^T &\neq 0, \\
        (Syn_{G_A}(e)| Syn_{G'_B}(e)| Syn_{{G'}_C^{(0)}}(e))^T&\neq 0. \end{align*}
    Note that if $Syn_{G_A}(e)\neq 0$, then $e$ is always detectable since each intermediate code includes the check operators from $G_A$.  Thus, we only need to consider the case where
    $Syn_{G_A}(e)=0$, and so we can assume that $(Syn_{G_B}(e)| Syn_{G_C^{(0)}}(e))^T\neq 0$ and $(Syn_{G'_B}(e)| Syn_{{G'}_C^{(0)}}(e))^T\neq 0$. 

    Let $P_e$ denote the probability that the error $e$ is undetectable in some intermediate code over the random choices of $U$, $V$, and $V'$. We divide $P_e$ into three parts. Let $A_e$ denote the event that $Syn_{G_C^{(1)}}(e)=0$, $B_e$ the event that $Syn_{{G'}_C^{(1)}}(e)=0$, and let $C_e$ denote the event that both $Syn_{G_C^{(1)}}(e)$ and $Syn_{{G_C}'^{(1)}}(e)$ are nonzero, yet $e$ becomes undetectable on some intermediate code during the transformation. Then $P_e\leq \Pr[A_e]+\Pr[B_e]+\Pr[C_e]$. We bound $\Pr[A_e]$, $\Pr[B_e]$, and $\Pr[C_e]$ separately. To bound $\Pr[A_e]$, note that $$Syn_{G_C^{(1)}}(e)=U(V\cdot Syn_{G_B}(e)+Syn_{G_C^{(0)}}(e)).$$ Since $U\in GL(\mathbb{F}_2,n-k)$, $A_e$ occurs if and only if $V\cdot Syn_{G_B}(e)+Syn_{G_C^{(0)}}(e)=0$. If $Syn_{G_B}(e)=0$, it must be the case that $Syn_{G_C^{(0)}}(e)\neq 0$, and so $Syn_{G_C^{(1)}}(e)\neq 0$; otherwise $Syn_{G_B}(e)\neq0$ and $V\cdot Syn_{G_B}(e)+Syn_{G_C^{(0)}}(e)$ is uniformly random over $\{0,1\}^{|G_C|}$.  In either case, we have
            $$\Pr[A_e] \leq 2^{-|G_C|}.$$
            Repeating the same argument shows that $\Pr[B_e]\leq 2^{-|G_C|}$ as well.  To bound $\Pr[C_e]$, define
\begin{align*}
v&=V\cdot Syn_{G_B}(e)+Syn_{G_C^{(0)}}(e), \\ w&=V'\cdot Syn_{G'_B}(e)+Syn_{{G'}_C^{(0)}}(e).
\end{align*}
Since $Uv, (U^{-1})^Tw\neq 0$, $e$ will be detectable during steps \ref{10} and \ref{12}, and so $C_e$ occurs only if $e$ becomes undetectable during step \ref{11}.  Specifically, it must be that $Syn_{G_A}(e)=0, Syn_{\overline{G}_B}(e)=0$, and the last $1$ in the vector $Uv$ occurs before the first $1$ in the vector $(U^{-1})^Tw$.  This is because we are sequentially replacing the check operators of $G$ with the check operators of $G'$, and so an error becomes undetectable for some intermediate code only if we produce some zero syndrome during this sequence of substitutions. By Lemma~\ref{lem:1}, for two nonzero vectors $v,w\in\{0,1\}^{|G_C|}$, the probability that the last $1$ in $Uv$ comes before the first $1$ in $(U^{-1})^Tw$ is bounded by $(|G_C|-1)\cdot 2^{-|G_C|}$.

 Summing these three terms, we have $P_e\leq (|G_C|+1)\cdot 2^{-|G_C|}$.  Taking a union bound, the probability $P$ that any of the intermediate codes fail to detect any error of weight less than $d$ is upper bounded by $$P \leq \sum_{e:|e|< d}P_e \leq |\{e:|e|< d\}|\cdot (|G_C|+1)\cdot 2^{-|G_C|}.$$  Taking a Chernoff bound, we get that this is in turn upper bounded as $$P \leq 4^{n+m}\cdot e^{-D(\frac{d-1}{n+m}||\frac{3}{4})(n+m)}\cdot  (|G_C|+1)\cdot 2^{-|G_C|}$$ where $D(\cdot||\cdot)$ is the KL-divergence.  By the quantum singleton bound, we can assume $\frac{d-1}{n+m}<\frac{d-1}{n}<\frac{3}{4}$. Furthermore, by Lemma \ref{lem:2}, $|G_C|$ is given by rank$(G^TBG')$, which is at least $m$.  So the probability of failure can be further upper bounded by
    $$ P\leq 4^{n+m}\cdot e^{-D(\frac{d-1}{n+m}\Vert\frac{3}{4})(n+m)}\cdot  (m+1)\cdot 2^{-m}.$$
    It suffices to choose $m$ such that the above quantity is upper bounded by $\epsilon$ in order to achieve a high probability of success. In particular, the case $\varepsilon=1$ upper bounds the minimum number of ancilla qubits required for a fault-tolerant transformation. By Lemma~\ref{lem:3} we observe that taking
    $$m=O(d\log\frac{n}{d}+\log\frac1\varepsilon)$$
    is sufficient for the rSRA scheme to succeed with probability $1-\epsilon$.

\end{proof}

\section{Discussion}
\label{sec:discussion}
Theorem \ref{thm:1} shows that with high probability, the rSRA will produce a transversal circuit with intermediate codes that have distances \emph{at least} the minimum of the distances of the extremal codes.  It is important to note that this does \emph{not} necessarily imply fault-tolerance.  The reason is because, when measuring $g'$, the randomness in the outcome prevents us from using that syndrome bit during error-correction. More specifically, consider the following two scenarios.
\begin{enumerate}
\item We project onto the $(+1)$-eigenspace of $g'$.
\item We project onto the $(-1)$-eigenspace of $g'$ and simultaneously experience an error that anticommutes with only $g'$. 
\end{enumerate}

Then we cannot distinguish these two scenarios using only our syndrome bits, and so cannot correct the resulting error.  More generally, we can cast the property required for fault-tolerance in terms of subsystem codes.  For every conversion between adjacent codes, we consider the subsystem code with a single gauge degree of freedom corresponding to gauge operators $g'$ and $g$.  Then the resulting conversion will be $t$-fault-tolerant precisely when the resulting subsystem code has distance $2t+1$.  This is because the redundant syndrome information can diagnose errors without the syndrome bit associated to $g'$, and so ensure that we project onto the correct eigenspace.  For this reason, additional techniques may be required to achieve fault-tolerance using the rSRA, such as error-detection on the ancilla.  We leave this to future work.

These techniques contrast with recent results from \cite{Yoder:2016}, where it was shown that pieceable fault-tolerance offers generic \emph{fault-tolerant} code switching between stabilizer codes subject to certain constraints.  However, their techniques require that the codes are nondegenerate and have some set of native fault-tolerant Clifford gates, allowing a fault-tolerant SWAP gate \emph{between} different codes.  One could also consider preparing a second code state and using logical teleportation to achieve a fault-tolerant mapping \cite{Brun:2015}.

Practically, on small examples, one finds that often \emph{no} ancilla qubits are required to find a distance-preserving circuit, which is desirable as the resulting circuit may then be fault-tolerant.  In general, this can be attributed to a coarse accounting of $|G_C|$ in terms of the number $m$ of ancilla qubits.  In most cases, $N(S) \cap N(S')$ will be small, and so the ancilla will be superfluous.  

Moreover, the multi-qubit gate complexity of the algorithm is $\sum_{P \in \{\overline{g_i}\} \cup G_B \cup G_C} |P|$, so that choosing a low weight generating set is ideal for reducing the complexity of the code switching circuit.  For this reason, LDPC codes might provide more efficient code switching circuits, although preserving the distance may depend on choosing a high weight set of generators.

This algorithm derives its usefulness from its generality.  For specific code switching examples, it may be profitable to modify the circuit using the rSRA as a template, augmented with a larger class of fault-tolerant manipulations such as local Clifford gates, in order to search for a \emph{fault-tolerant} mapping.  For large code sizes, the use of high-weight Shor-style measurements is limiting as it requires large verified CAT states.  Thus, this technique may be most useful as a step in a concatenated scheme, or simply as a search ansatz.

One subtlety about the rSRA is that, while it outputs a distance-preserving circuit switching between two codes with high probability, this is difficult to check.  This follows from the difficulty of computing the minimum distance of a generic error-correcting code, which is an NP-hard problem in general \cite{Dumer:2003}.  Indeed, even when restricting to a particular distance, this check remains extremely costly.  This poses a barrier to derandomizing the algorithm, which would be one desirable avenue for future improvement.  

Another such improvement would be to minimize overhead.  One could imagine taking a random local clifford transformation in order to increase the size of $G_C$, rather than introducing ancilla.  Such a strategy would be interesting since locally equivalent codes have nearly identical properties.  Of course, modifying the algorithm to ensure fault-tolerance is the most important improvement.  

If it is true that one can always choose locally equivalent representatives for which the rSRA provides a distance-preserving conversion without ancilla, this would suggest that all error-protected information in stabilizer codes is, in some sense, ``transversally equivalent''.  This contrasts with the diverse set of equivalence classes of locally unitarily equivalent codes, which can be identified as distinct submanifolds of Grassmanians.  Indeed, it may be of conceptual interest to interpret these upper-bounds in a broader framework of fault-tolerance, such as the one investigated in \cite{Gottesman:2013}.  

Similarly, the generality of the rSRA provides an aesthetically nice interpretation of error-protected information.  It suggests that, with the addition of some minimal overhead, any stabilizer error-protected encoding of information is indeed ``transversally equivalent'' to any other. 
\section{Acknowledgments}
The authors would like to thank Ted Yoder for pointing out a critical oversight in an earlier version of this draft.  They would also like to thank Rui Chao, Fang Zhang, Kevin Sung, Daniel Minahan, Yaoyun Shi, and anonymous QIP referees for their valuable comments and suggestions.  This research was supported in part by NSF grant 1717523.

\appendix

\section{Technical lemmas}

\begin{Lemma}\label{lem:1}
    Let $v,w\in\{0,1\}^n\setminus\{0\}$ and $U\in_r GL(\mathbb{F}_2, n)$. Let $i_0=\max\{i:(U\cdot v)_i=1\}$ and $i_1=\min\{i:((U^{-1})^T\cdot w)_i=1\}$. Then,
    $$\Pr[i_0<i_1]\leq (n-1)\cdot 2^{-n}.$$
\end{Lemma}
\begin{proof}
    Let $\langle \cdot,\cdot\rangle$ be the dot product over $\mathbb{F}_2$. Note that $\langle v,w\rangle = \langle U\cdot v, (U^{-1})^T\cdot w\rangle$. If $\langle v,w\rangle =1$, then there must be at least one entry where both $U\cdot v$ and $(U^{-1})^{T}\cdot w$ are $1$ for whichever $U$ we choose, and so $\Pr[i_0<i_1]=0$. Therefore we only need to consider the case in which $\langle v,w\rangle=0$.

    Consider the action of $GL(\mathbb{F}_2, n)$ on $A$ defined by $U (v,w)\rightarrow (U\cdot v, (U^{-1})^T\cdot w)$, where $A=\{(v,w)|v,w\in \{0,1\}^n\setminus\{0\}, \langle v,w\rangle = 0\}$. We show that the action is transitive by showing that for all such pairs $(v,w)$, there always exists a $U$ sending $(e_1,e_n)$ to $(v,w)$, where $e_1, e_n$ are $(1,0,\ldots, 0)$ and $(0,0,0,\ldots, 1)$, respectively. Given such a $(v,w)$, first extend $v$ to a basis for $w^{\bot}$, say $(u_1=v,u_2,\ldots,
    u_{n-1})$, and then extend it to the whole space by adding in $u_n$. We claim that $U=(u_1,\cdots, u_n)$ is the desired matrix. It is sufficient to show that the last column $w'$ of $(U^{-1})^T$  is exactly $w$. We have $U^T w' = e_n$ given that $U^T(U^{-1})^T=I$, and that $U^Tw=e_n$ by construction of $U$. Then, since $U$ is invertible, $w=w'$.

    A uniformly random distribution over invertible $U$ then induces a uniformly random distribution over $A$. Then $\Pr[i_0<i_1]$ can then be bounded by counting the number of such pairs in $A$:
    \begin{align*}
        \Pr[i_0<i_1] &=\frac{\sum_{i_0=1}^n 2^{i_0-1}(2^{n-i_0}-1)}{(2^n-1)(2^{n-1}-1)}\\
        &=\frac{(n-2)2^{n-1}+1}{(2^n-1)(2^{n-1}-1)}\\
        &\leq (n-1)\cdot2^{-n}
    \end{align*}
    when $n\geq 2$. Note that $|A|=0$ when $n=1$, so $\Pr[i_0<i_1]\leq (n-1)\cdot 2^{-n}$ holds for all $n\geq 0$.
\end{proof}
\begin{Lemma}
    Let $G_A,G_B,G_C$ and $G_A,G'_B,G'_C$ be the matrices defined up to step \ref{5} in the rSRA scheme. The commutativity matrix $H=G_C^TBG_C'$ is invertible, and its dimension is $|G_C|$, with $|G_C|\geq m$.
    \label{lem:2}
\end{Lemma}
\begin{proof}
 For the two  codes $\hat{S}$ and $\hat{S}'$, take arbitrary generator matrices $G$,$G'$ and define $H'=G^TBG'$. Note that any two choices of generator matrices for the same code differ by an invertible row transformation, so the rank of $H'$ is invariant under different choices of the generator matrices. In particular, letting $G=(G_A|G_B|G_C), G'=(G'_A|G'_B|G'_C)$, we have
 $$H'=\begin{bmatrix}0& & \\ &0&\\ &&H\end{bmatrix}.$$

Note that $\mathrm{rank}(H) = |G_C|$, or else there would exist a combination of the rows of ${G_C'}^T$ that are orthogonal to all the columns of $G_C$.  Since all the vectors in $G_C'$ are already orthogonal to $G_A$ and $G_B$ by definition, this cannot happen as no vector in $G_C'$ lies in $\mathcal{N}(S)$. The same argument applies to $G'_C$ as well. Therefore $H$ is invertible, with $\mathrm{rank}(H')=\mathrm{rank}(H)=|G_C|$, and is independent of the choice of $G_C$.

To show that $|G_C|>m$, take $\bar{G}_C=(I^{\otimes n}\otimes Z\otimes I^{\otimes m-1},\ldots, I^{\otimes n+m-1}\otimes Z )$ and $\bar{G}'_C=(I^{\otimes n}\otimes X\otimes I^{\otimes m-1},\ldots, I^{\otimes n+m-1}\otimes X )$, each of size $m$. By extending them to generator matrices $\bar{G}$ and $\bar{G}'$ for $\hat{S}$ and $\hat{S}'$ respectively, we get a commutativity matrix $\bar{H}$ with an invertible submatrix of size $m\times m$, namely
$$\bar{G}_C^TB\bar{G}_C'=I_{m},$$
and so $\mathrm{rank}(\bar{H})=|G_C|\geq m$.
\end{proof}

\begin{Lemma}
    For $D(\cdot||\cdot)$ the KL-divergence, let $$P(n,m,d) = 4^{n+m}e^{-D(\frac{d}{n+m}\Vert \frac34)(n+m)}\cdot (m+1)\cdot 2^{-m}.$$ Then $P<\varepsilon$ for some $m = O(d\log\frac{n}{d}+\log \frac1\varepsilon)$.
    \label{lem:3}
\end{Lemma}

\begin{proof}
    Let $\alpha = m/n$. Then $P(n,m,d)<\varepsilon$ can be rewritten as
    \begin{align*}f(n,m,d) :&=\log\frac{P(n,m,d)}{\varepsilon} \\ &=\log \frac{m+1}{\varepsilon}+n\biggl((2+\alpha)\log 2\\ &-(1+\alpha)D\biggl(\frac{d}{n(1+\alpha)}\Vert \frac{3}{4}\biggr)\biggr)<0.\end{align*}
    We first compute the dominant term, i.e.\ the $\alpha$ such that $$\left(2+\alpha\right)\ln 2-(1+\alpha)D\left(\frac{d}{n(1+\alpha)}\Vert \frac{3}{4}\right)=0.$$  Doing this we obtain
    \begin{align*}
        (2-\frac{\alpha}{1+\alpha})\ln 2&=D\left(\frac{d}{n(1+\alpha)}\Vert \frac{3}{4}\right)\\
        (2-\frac{\alpha}{1+\alpha})\ln 2&\geq2\ln 2+\frac {d}{n(1+\alpha)}\left(\ln\frac{d}{3n(1+\alpha)}-1\right)\\
        \alpha n &\leq\frac d{\ln 2}\left(\ln\frac{3n(1+\alpha)}{d}+1\right)\\
        m &\leq \frac{1}{\ln 2} d(\log\frac nd + (1 + \ln3)),
    \end{align*}
    where we have used convexity of $D(p\Vert q)-p\ln p$ with respect to $p$. Letting $\tilde{\alpha}$ denote the solution to 
        $\left(2+\alpha\right)\ln 2-(1+\alpha)D\left(\frac{d}{n(1+\alpha)}\Vert \frac{3}{4}\right)=0$, we have that $\tilde{m}:=\tilde{\alpha}n=O(d+d\log\frac nd)$.

     We now have that $f(n,\tilde{m}, d)=\log \frac{\tilde{m}+1}{\epsilon}$. Taking the derivative of $f$ with respect to $m$, for all $\alpha>\tilde{\alpha}$ we have
    \begin{align*}
        &\frac{\partial f(n,m,d)}{\partial m} = \\ &=\frac{1}{m+1} -D(\frac{d}{(n+m)}\Vert\frac{3}{4})-(m+n)\frac{\partial D(\frac{d}{n+m}\Vert\frac{3}{4})}{\partial m}\\
        &\leq \frac{1}{m+1}-\frac{2+\tilde{\alpha}}{1+\tilde{\alpha}}\log 2+\frac{d}{m+n}\left(\log \frac{d}{3(m+n-d)}\right)\\
        &\leq \frac{1}{m+1}-\frac{1}{1+\tilde{\alpha}}\log 2+\frac{d}{n+m}\left(\log \frac{d}{3(n+m-d)}\right)\\
        &\leq -\frac{1}{1+\tilde{\alpha}}\ln 2+0.1
    \end{align*}
    for $m\geq 10$. For fixed $n$, $\tilde{\alpha}$ is monotonically increasing as a function of $d$. By the quantum singleton bound, $\frac{d-1}{n}<\frac{1}{2}$, and $\tilde{\alpha}<3$ even in this case. Therefore $\frac{\partial f(n,m,d)}{\partial m}\leq -0.05$ when $m\geq 10$, so taking $$m=\tilde{m}+20\log \frac{\tilde{m}+1}\epsilon+O(1)=O(d\log\frac{n}{d} + \log\frac1\varepsilon)$$ suffices to make $f(n,m,d)<0$.
\end{proof}

\section{Fault-tolerant measurement} \label{Shor}
For completeness, we include a code switching circuit between adjacent codes, using Shor-style measurement \cite{Shor:1996}.  We assume access to a collection of verified CAT states.  Let $g = P_1 \otimes \ldots \otimes P_n$ and $g' = P_1' \otimes \ldots \otimes P_n'
$.  The measurements are done on the supports of $g$ and $g'$.  To make the diagram simpler, we suppose that the supports include qubits $1,2$, and $n$.  Then the circuit obtained from the SRA to convert from the code with stabilizer $g'$ to the adjacent code with stabilizer $g$ is given by the following.

\begin{figure}[h]
\hspace*{-1.1cm}  
\includegraphics[scale = 0.42]{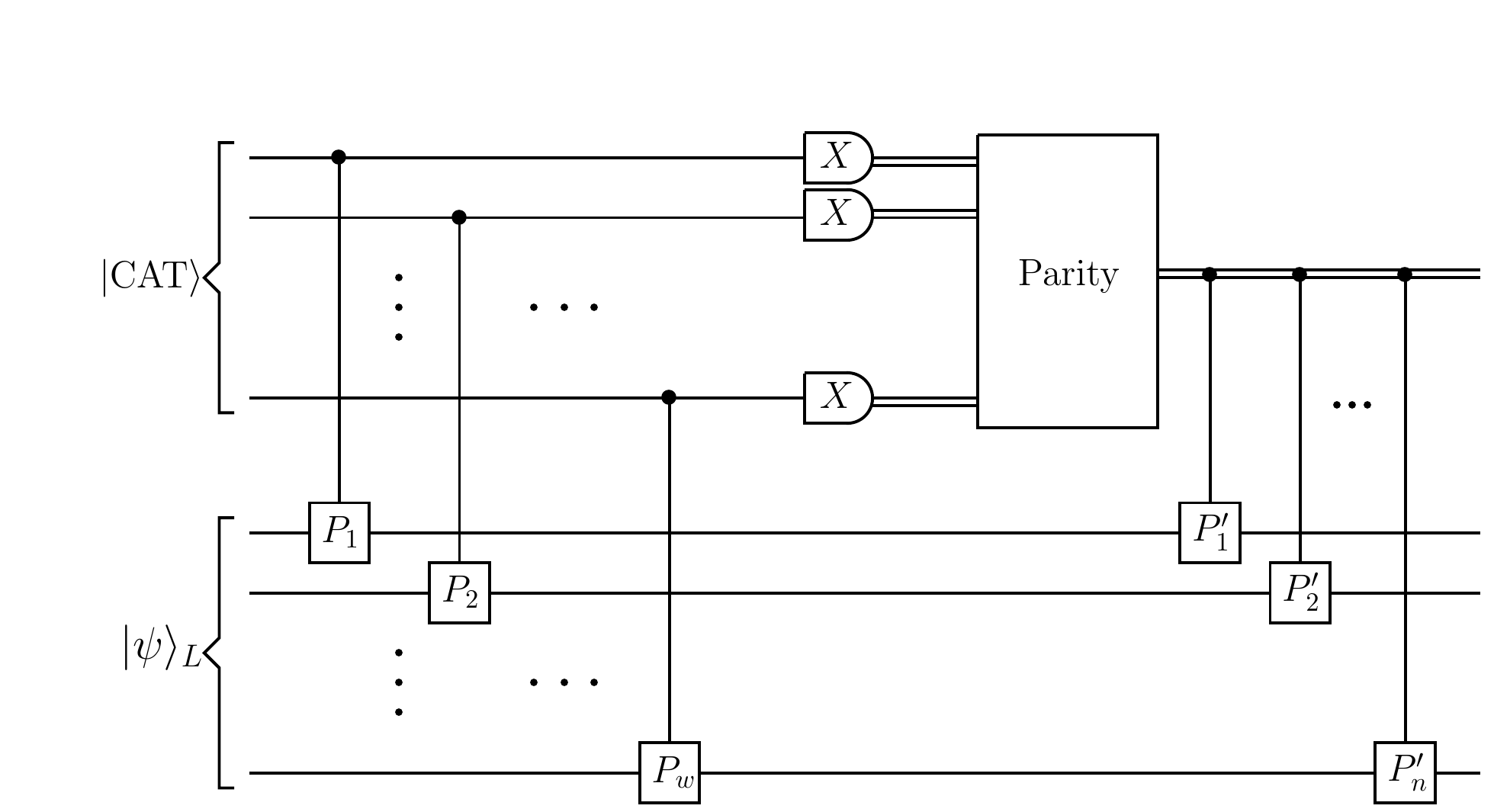}
\caption{A generic circuit switching between adjacent codes using Shor-style measurement.}\end{figure}


\newpage

\bibliography{bibliography}

\begin{thebibliography}{19}%
\makeatletter
\providecommand \@ifxundefined [1]{%
 \@ifx{#1\undefined}
}%
\providecommand \@ifnum [1]{%
 \ifnum #1\expandafter \@firstoftwo
 \else \expandafter \@secondoftwo
 \fi
}%
\providecommand \@ifx [1]{%
 \ifx #1\expandafter \@firstoftwo
 \else \expandafter \@secondoftwo
 \fi
}%
\providecommand \natexlab [1]{#1}%
\providecommand \enquote  [1]{``#1''}%
\providecommand \bibnamefont  [1]{#1}%
\providecommand \bibfnamefont [1]{#1}%
\providecommand \citenamefont [1]{#1}%
\providecommand \href@noop [0]{\@secondoftwo}%
\providecommand \href [0]{\begingroup \@sanitize@url \@href}%
\providecommand \@href[1]{\@@startlink{#1}\@@href}%
\providecommand \@@href[1]{\endgroup#1\@@endlink}%
\providecommand \@sanitize@url [0]{\catcode `\\12\catcode `\$12\catcode
  `\&12\catcode `\#12\catcode `\^12\catcode `\_12\catcode `\%12\relax}%
\providecommand \@@startlink[1]{}%
\providecommand \@@endlink[0]{}%
\providecommand \url  [0]{\begingroup\@sanitize@url \@url }%
\providecommand \@url [1]{\endgroup\@href {#1}{\urlprefix }}%
\providecommand \urlprefix  [0]{URL }%
\providecommand \Eprint [0]{\href }%
\providecommand \doibase [0]{http://dx.doi.org/}%
\providecommand \selectlanguage [0]{\@gobble}%
\providecommand \bibinfo  [0]{\@secondoftwo}%
\providecommand \bibfield  [0]{\@secondoftwo}%
\providecommand \translation [1]{[#1]}%
\providecommand \BibitemOpen [0]{}%
\providecommand \bibitemStop [0]{}%
\providecommand \bibitemNoStop [0]{.\EOS\space}%
\providecommand \EOS [0]{\spacefactor3000\relax}%
\providecommand \BibitemShut  [1]{\csname bibitem#1\endcsname}%
\let\auto@bib@innerbib\@empty
\bibitem [{\citenamefont {Eastin}\ and\ \citenamefont
  {Knill}(2009)}]{Eastin:2008}%
  \BibitemOpen
  \bibfield  {author} {\bibinfo {author} {\bibfnamefont {B.}~\bibnamefont
  {Eastin}}\ and\ \bibinfo {author} {\bibfnamefont {E.}~\bibnamefont {Knill}},\
  }\href@noop {} {\enquote {\bibinfo {title} {Restrictions on transversal
  encoded quantum gate sets},}\ } (\bibinfo {year} {2009}),\ \bibinfo {note}
  {{P}hys. Rev. Lett. 102, 110502}\BibitemShut {NoStop}%
\bibitem [{\citenamefont {Zeng}\ \emph {et~al.}(2011)\citenamefont {Zeng},
  \citenamefont {Cross},\ and\ \citenamefont {Chuang}}]{Zeng:2007}%
  \BibitemOpen
  \bibfield  {author} {\bibinfo {author} {\bibfnamefont {B.}~\bibnamefont
  {Zeng}}, \bibinfo {author} {\bibfnamefont {A.}~\bibnamefont {Cross}}, \ and\
  \bibinfo {author} {\bibfnamefont {I.~L.}\ \bibnamefont {Chuang}},\
  }\href@noop {} {\enquote {\bibinfo {title} {Transversality versus
  universality for additive quantum codes},}\ } (\bibinfo {year} {2011}),\
  \bibinfo {note} {{I}EEE Transactions on Information Theory, Volume: 57,
  Issue: 9, 6272 - 6284}\BibitemShut {NoStop}%
\bibitem [{\citenamefont {Newman}\ and\ \citenamefont
  {Shi}(2017)}]{Newman:2017}%
  \BibitemOpen
  \bibfield  {author} {\bibinfo {author} {\bibfnamefont {M.}~\bibnamefont
  {Newman}}\ and\ \bibinfo {author} {\bibfnamefont {Y.}~\bibnamefont {Shi}},\
  }\href@noop {} {\enquote {\bibinfo {title} {Limitations on transversal
  computation through quantum homomorphic encryption},}\ } (\bibinfo {year}
  {2017}),\ \bibinfo {note} {https://arxiv.org/abs/1704.07798}\BibitemShut
  {NoStop}%
\bibitem [{\citenamefont {Fowler}\ \emph {et~al.}(2013)\citenamefont {Fowler},
  \citenamefont {Devitt},\ and\ \citenamefont {Jones}}]{Fowler:2013}%
  \BibitemOpen
  \bibfield  {author} {\bibinfo {author} {\bibfnamefont {A.~G.}\ \bibnamefont
  {Fowler}}, \bibinfo {author} {\bibfnamefont {S.~J.}\ \bibnamefont {Devitt}},
  \ and\ \bibinfo {author} {\bibfnamefont {C.}~\bibnamefont {Jones}},\
  }\href@noop {} {\enquote {\bibinfo {title} {Surface code implementation of
  block code state distillation},}\ } (\bibinfo {year} {2013}),\ \bibinfo
  {note} {{S}cientific Reports 3, 1939}\BibitemShut {NoStop}%
\bibitem [{\citenamefont {Bravyi}\ and\ \citenamefont
  {Haah}(2012)}]{Bravyi:2012}%
  \BibitemOpen
  \bibfield  {author} {\bibinfo {author} {\bibfnamefont {S.}~\bibnamefont
  {Bravyi}}\ and\ \bibinfo {author} {\bibfnamefont {J.}~\bibnamefont {Haah}},\
  }\href@noop {} {\enquote {\bibinfo {title} {Magic state distillation with low
  overhead},}\ } (\bibinfo {year} {2012}),\ \bibinfo {note} {{P}hys. Rev. A 86,
  052329}\BibitemShut {NoStop}%
\bibitem [{\citenamefont {Paetznick}\ and\ \citenamefont
  {Reichardt}(2013)}]{Paetznick:2013}%
  \BibitemOpen
  \bibfield  {author} {\bibinfo {author} {\bibfnamefont {A.}~\bibnamefont
  {Paetznick}}\ and\ \bibinfo {author} {\bibfnamefont {B.~W.}\ \bibnamefont
  {Reichardt}},\ }\href@noop {} {\enquote {\bibinfo {title} {Universal
  fault-tolerant quantum computation with only transversal gates and error
  correction},}\ } (\bibinfo {year} {2013}),\ \bibinfo {note} {{P}hys. Rev.
  Lett. 111, 090505 (2013)}\BibitemShut {NoStop}%
\bibitem [{\citenamefont {Bombin}(2015)}]{Bombin:2013}%
  \BibitemOpen
  \bibfield  {author} {\bibinfo {author} {\bibfnamefont {H.}~\bibnamefont
  {Bombin}},\ }\href@noop {} {\enquote {\bibinfo {title} {Gauge color codes:
  Optimal transversal gates and gauge fixing in topological stabilizer
  codes},}\ } (\bibinfo {year} {2015}),\ \bibinfo {note} {{N}ew J. Phys. 17
  (2015) 083002}\BibitemShut {NoStop}%
\bibitem [{\citenamefont {Yoder}\ \emph {et~al.}(2016)\citenamefont {Yoder},
  \citenamefont {Takagi},\ and\ \citenamefont {Chuang}}]{Yoder:2016}%
  \BibitemOpen
  \bibfield  {author} {\bibinfo {author} {\bibfnamefont {T.~J.}\ \bibnamefont
  {Yoder}}, \bibinfo {author} {\bibfnamefont {R.}~\bibnamefont {Takagi}}, \
  and\ \bibinfo {author} {\bibfnamefont {I.~L.}\ \bibnamefont {Chuang}},\
  }\href@noop {} {\enquote {\bibinfo {title} {Universal fault-tolerant gates on
  concatenated stabilizer codes},}\ } (\bibinfo {year} {2016}),\ \bibinfo
  {note} {{P}hys. Rev. X 6, 031039 (2016)}\BibitemShut {NoStop}%
\bibitem [{\citenamefont {Yoder}(2017)}]{Yoder:2017}%
  \BibitemOpen
  \bibfield  {author} {\bibinfo {author} {\bibfnamefont {T.~J.}\ \bibnamefont
  {Yoder}},\ }\href@noop {} {\enquote {\bibinfo {title} {Universal
  fault-tolerant quantum computation with {B}acon-{S}hor codes},}\ } (\bibinfo
  {year} {2017}),\ \bibinfo {note} {arXiv:1705.01686}\BibitemShut {NoStop}%
\bibitem [{\citenamefont {Hill}\ \emph {et~al.}(2013)\citenamefont {Hill},
  \citenamefont {Fowler}, \citenamefont {Wang},\ and\ \citenamefont
  {Hollenberg}}]{Hill:2013}%
  \BibitemOpen
  \bibfield  {author} {\bibinfo {author} {\bibfnamefont {C.~D.}\ \bibnamefont
  {Hill}}, \bibinfo {author} {\bibfnamefont {A.~G.}\ \bibnamefont {Fowler}},
  \bibinfo {author} {\bibfnamefont {D.~S.}\ \bibnamefont {Wang}}, \ and\
  \bibinfo {author} {\bibfnamefont {L.~C.}\ \bibnamefont {Hollenberg}},\
  }\href@noop {} {\enquote {\bibinfo {title} {Fault-tolerant quantum error
  correction code conversion},}\ } (\bibinfo {year} {2013}),\ \bibinfo {note}
  {{Q}uantum Inf. Comput. 13, 439Ð451}\BibitemShut {NoStop}%
\bibitem [{\citenamefont {Anderson}\ \emph {et~al.}(2014)\citenamefont
  {Anderson}, \citenamefont {Duclos-Cianci},\ and\ \citenamefont
  {Poulin}}]{Anderson:2014}%
  \BibitemOpen
  \bibfield  {author} {\bibinfo {author} {\bibfnamefont {J.~T.}\ \bibnamefont
  {Anderson}}, \bibinfo {author} {\bibfnamefont {G.}~\bibnamefont
  {Duclos-Cianci}}, \ and\ \bibinfo {author} {\bibfnamefont {D.}~\bibnamefont
  {Poulin}},\ }\href@noop {} {\enquote {\bibinfo {title} {Fault-tolerant
  conversion between the {S}teane and {R}eed-{M}uller quantum codes},}\ }
  (\bibinfo {year} {2014}),\ \bibinfo {note} {{P}hys. Rev. Lett. 113,
  080501}\BibitemShut {NoStop}%
\bibitem [{\citenamefont {Brun}\ \emph {et~al.}(2015)\citenamefont {Brun},
  \citenamefont {Zheng}, \citenamefont {Hsu}, \citenamefont {Job},\ and\
  \citenamefont {Lai}}]{Brun:2015}%
  \BibitemOpen
  \bibfield  {author} {\bibinfo {author} {\bibfnamefont {T.~A.}\ \bibnamefont
  {Brun}}, \bibinfo {author} {\bibfnamefont {Y.-C.}\ \bibnamefont {Zheng}},
  \bibinfo {author} {\bibfnamefont {K.-C.}\ \bibnamefont {Hsu}}, \bibinfo
  {author} {\bibfnamefont {J.}~\bibnamefont {Job}}, \ and\ \bibinfo {author}
  {\bibfnamefont {C.-Y.}\ \bibnamefont {Lai}},\ }\href@noop {} {\enquote
  {\bibinfo {title} {Teleportation-based fault-tolerant quantum computation in
  multi-qubit large block codes},}\ } (\bibinfo {year} {2015}),\ \bibinfo
  {note} {arXiv:1504.03913}\BibitemShut {NoStop}%
\bibitem [{\citenamefont {Bombin}\ and\ \citenamefont
  {Martin-Delgado}(2009)}]{Bombin:2007}%
  \BibitemOpen
  \bibfield  {author} {\bibinfo {author} {\bibfnamefont {H.}~\bibnamefont
  {Bombin}}\ and\ \bibinfo {author} {\bibfnamefont {M.}~\bibnamefont
  {Martin-Delgado}},\ }\href@noop {} {\enquote {\bibinfo {title} {Quantum
  measurements and gates by code deformation},}\ } (\bibinfo {year} {2009}),\
  \bibinfo {note} {{J}. Phys. A: Math. Theor. 42}\BibitemShut {NoStop}%
\bibitem [{\citenamefont {Nautrup}\ \emph {et~al.}(2017)\citenamefont
  {Nautrup}, \citenamefont {Friis},\ and\ \citenamefont
  {Briegel}}]{Nautrup:2017}%
  \BibitemOpen
  \bibfield  {author} {\bibinfo {author} {\bibfnamefont {H.~P.}\ \bibnamefont
  {Nautrup}}, \bibinfo {author} {\bibfnamefont {N.}~\bibnamefont {Friis}}, \
  and\ \bibinfo {author} {\bibfnamefont {H.~J.}\ \bibnamefont {Briegel}},\
  }\href@noop {} {\enquote {\bibinfo {title} {Fault-tolerant interface between
  quantum memories and quantum processors},}\ } (\bibinfo {year} {2017}),\
  \bibinfo {note} {{N}ature Communications 8, Article number: 1321}\BibitemShut
  {NoStop}%
\bibitem [{\citenamefont {Colladay}\ and\ \citenamefont
  {Mueller}(2017)}]{Colladay:2017}%
  \BibitemOpen
  \bibfield  {author} {\bibinfo {author} {\bibfnamefont {K.~R.}\ \bibnamefont
  {Colladay}}\ and\ \bibinfo {author} {\bibfnamefont {E.~J.}\ \bibnamefont
  {Mueller}},\ }\href@noop {} {\enquote {\bibinfo {title} {Rewiring stabilizer
  codes},}\ } (\bibinfo {year} {2017}),\ \bibinfo {note}
  {arXiv:1707.09403}\BibitemShut {NoStop}%
\bibitem [{\citenamefont {Aliferis}\ \emph {et~al.}(2006)\citenamefont
  {Aliferis}, \citenamefont {Gottesman},\ and\ \citenamefont
  {Preskill}}]{Aliferis:2006}%
  \BibitemOpen
  \bibfield  {author} {\bibinfo {author} {\bibfnamefont {P.}~\bibnamefont
  {Aliferis}}, \bibinfo {author} {\bibfnamefont {D.}~\bibnamefont {Gottesman}},
  \ and\ \bibinfo {author} {\bibfnamefont {J.}~\bibnamefont {Preskill}},\
  }\href@noop {} {\enquote {\bibinfo {title} {Quantum accuracy threshold for
  concatenated distance-3 codes},}\ } (\bibinfo {year} {2006}),\ \bibinfo
  {note} {{Q}uantum Inf. Comput. 6, 97Ð165}\BibitemShut {NoStop}%
\bibitem [{\citenamefont {Dumer}\ \emph {et~al.}(2003)\citenamefont {Dumer},
  \citenamefont {Micciancio},\ and\ \citenamefont {Sudan}}]{Dumer:2003}%
  \BibitemOpen
  \bibfield  {author} {\bibinfo {author} {\bibfnamefont {I.}~\bibnamefont
  {Dumer}}, \bibinfo {author} {\bibfnamefont {D.}~\bibnamefont {Micciancio}}, \
  and\ \bibinfo {author} {\bibfnamefont {M.}~\bibnamefont {Sudan}},\
  }\href@noop {} {\enquote {\bibinfo {title} {Hardness of approximating the
  minimum distance of a linear code},}\ } (\bibinfo {year} {2003}),\ \bibinfo
  {note} {{I}EEE Transactions on Information Theory, 49(1):22-37}\BibitemShut
  {NoStop}%
\bibitem [{\citenamefont {Gottesman}\ and\ \citenamefont
  {Zhang}(2013)}]{Gottesman:2013}%
  \BibitemOpen
  \bibfield  {author} {\bibinfo {author} {\bibfnamefont {D.}~\bibnamefont
  {Gottesman}}\ and\ \bibinfo {author} {\bibfnamefont {L.~L.}\ \bibnamefont
  {Zhang}},\ }\href@noop {} {\enquote {\bibinfo {title} {Fibre bundle framework
  for unitary quantum fault tolerance},}\ } (\bibinfo {year} {2013}),\ \bibinfo
  {note} {arXiv:1309.7062}\BibitemShut {NoStop}%
\bibitem [{\citenamefont {Shor}(1996)}]{Shor:1996}%
  \BibitemOpen
  \bibfield  {author} {\bibinfo {author} {\bibfnamefont {P.~W.}\ \bibnamefont
  {Shor}},\ }\href@noop {} {\enquote {\bibinfo {title} {Fault-tolerant quantum
  computation},}\ } (\bibinfo {year} {1996}),\ \bibinfo {note} {37th Symposium
  on Foundations of Computing, IEEE Computer Society Press, pp.
  56-65}\BibitemShut {NoStop}%
\end{thebibliography}%

\end{document}